\documentclass[11pt]{article}

\usepackage[letterpaper,margin=1in]{geometry}
\usepackage[T1]{fontenc}
\usepackage[utf8]{inputenc}
\usepackage{lmodern}
\usepackage{microtype}
\usepackage{amsmath,amssymb,amsthm,mathtools}
\usepackage{booktabs,tabularx,array}
\usepackage{enumitem}
\usepackage{float}
\usepackage{placeins}
\usepackage{xcolor}
\usepackage[numbers]{natbib}
\usepackage{xspace}
\usepackage{tikz}
\usetikzlibrary{arrows.meta,positioning,calc,fit,backgrounds,decorations.pathreplacing,shapes.geometric}
\usepackage{hyperref}
\usepackage{aliascnt}
\usepackage[nameinlink,noabbrev]{cleveref}

\definecolor{ravblue}{RGB}{25,74,120}
\definecolor{ravlight}{RGB}{243,247,251}
\definecolor{ravgreen}{RGB}{32,112,72}
\definecolor{ravorange}{RGB}{164,91,18}
\definecolor{ravgray}{RGB}{91,101,111}

\usepackage{pgfplots}
\pgfplotsset{compat=1.18}
\usetikzlibrary{calc,decorations.pathreplacing}
\usepgfplotslibrary{groupplots}

\pgfplotsset{
  paperplot/.style={
    tick label style={font=\footnotesize},
    label style={font=\footnotesize},
    title style={font=\footnotesize},
    legend style={font=\footnotesize, draw=none, fill=none},
    axis line style={black!70},
    tick style={black!70},
    every axis plot/.append style={line width=0.8pt},
  }
}

\hypersetup{
  pdftitle={Threshold Monotonicity of Sampford Apportionment},
  colorlinks=true,
  linkcolor=ravblue,
  citecolor=ravblue,
  urlcolor=ravblue
}

\usepackage{xcolor}
  \newcommand{\harisnew}[1]{\textcolor{black}{#1}}

\newcommand{\restatedhead}{}
\newtheorem*{restatedmainthmaux}{\restatedhead}

\newtheorem{theorem}{Theorem}[section]
\newaliascnt{lemma}{theorem}
\newtheorem{lemma}[lemma]{Lemma}
\aliascntresetthe{lemma}
\newaliascnt{proposition}{theorem}

\aliascntresetthe{proposition}
\newaliascnt{corollary}{theorem}
\newtheorem{corollary}[corollary]{Corollary}
\aliascntresetthe{corollary}
\theoremstyle{definition}
\newaliascnt{definition}{theorem}
\newtheorem{definition}[definition]{Definition}
\aliascntresetthe{definition}
\newaliascnt{example}{theorem}
\newtheorem{example}[example]{Example}
\aliascntresetthe{example}
\theoremstyle{remark}
\newaliascnt{remark}{theorem}
\newtheorem{remark}[remark]{Remark}
\aliascntresetthe{remark}

\crefname{equation}{Equation}{Equations}
\crefname{figure}{Figure}{Figures}
\crefname{table}{Table}{Tables}
\crefname{section}{Section}{Sections}
\crefname{subsection}{Section}{Sections}
\crefname{mainthm}{Theorem}{Theorems}
\Crefname{mainthm}{Theorem}{Theorems}
\crefname{theorem}{Theorem}{Theorems}
\crefname{lemma}{Lemma}{Lemmas}
\crefname{proposition}{Proposition}{Propositions}
\crefname{corollary}{Corollary}{Corollaries}
\crefname{definition}{Definition}{Definitions}
\crefname{example}{Example}{Examples}
\crefname{remark}{Remark}{Remarks}

\newcommand{\N}{\mathbb{N}}
\newcommand{\Nzero}{\mathbb{N}_0}
\newcommand{\R}{\mathbb{R}}
\newcommand{\E}{\mathbb{E}}
\newcommand{\SD}{\mathrm{SD}}
\newcommand{\Samp}{\mathrm{Samp}}

\setlist[itemize]{topsep=4pt,itemsep=2pt,parsep=1pt,leftmargin=1.5em}
\setlist[enumerate]{topsep=4pt,itemsep=3pt,parsep=1pt,leftmargin=1.7em}
\title{\textbf{Sampford Apportionment Satisfies \\Threshold Monotonicity}}
\author{%
  Haris Aziz\\
  {\small UNSW Sydney}\\
  {\small\texttt{haris.aziz@unsw.edu.au}}
  \and
  Simon Mackenzie\\
  {\small UNSW Sydney}\\
  {\small\texttt{simon.william.mackenzie@gmail.com}}
  \and
  Mashbat Suzuki\\
  {\small UNSW Sydney}\\
  {\small\texttt{mashbat.suzuki@unsw.edu.au}}%
}
\date{}

\begin{document}
\maketitle

\begin{abstract}
Apportionment distributes a fixed number of legislative seats among political
parties in proportion to their vote shares. Randomization can satisfy quota in
every realization and exact proportionality in expectation, but these
requirements determine only the parties' marginal seat distributions. We study
threshold monotonicity, which requires that when every standard quota inside a
coalition weakly increases and every standard quota outside it weakly decreases,
the coalition's new seat total stochastically dominates its old one. We prove
that Sampford apportionment satisfies threshold monotonicity, resolving the
threshold-monotonicity conjecture of Correa et al. (2024).
Complementing this result, we prove that no quota-respecting, ex-ante
proportional method satisfies pairwise threshold monotonicity even for disjoint
coalitions when there are at least seven parties. This impossibility applies
whether or not the method has full support.
\end{abstract}


\section{Introduction}
\label{sec:introduction}

An apportionment problem asks how a house of size $h$ should be divided among
$n$ political parties in proportion to their vote totals. If party $i$ is
entitled to the standard quota $q_i$, integrality generally prevents it from
receiving exactly $q_i$ seats. A classical requirement is \emph{quota}: party
$i$ must receive either $\lfloor q_i\rfloor$ or $\lceil q_i\rceil$ seats.
Deterministic methods face a classical conflict between quota and the most
natural population-monotonicity requirement \citep{BaYo01a}. Randomization
avoids this particular
impasse. As observed by \citet{Grim04a}, a method may satisfy quota in every
realization while awarding each party exactly its standard quota in expectation. 

Exact expectations describe only individual marginals, while political
questions are often joint. A coalition may care about a majority, a blocking
minority, or another important threshold.
Two randomized methods with the same marginal
inclusion probabilities can assign very different probabilities to these
events. More troublingly, a coalition whose parties all become proportionally
stronger can become \emph{less} likely to reach a given threshold if the
rounding decisions are correlated in the wrong way. Such a reversal can also
create strategic incentives for voters who evaluate outcomes through the total
number of seats won by an approved set of parties.
To exclude this behaviour, \citet{CGST+24a,CGST+26a} introduced \emph{threshold
monotonicity}. In its quota-based form, the axiom requires the total number of
seats received by a reinforced coalition to increase in stochastic dominance.
This is the relevant comparison for every nondecreasing utility of
the coalition's seat total, and hence simultaneously controls all threshold
events.

\begin{example}\label{example:AB}
The following example explains why monotonicity should concern \emph{seat
thresholds}, not merely which parties are rounded up.  Consider a house of
 four seats and parties $A,B,C$, where $A$ and $B$ form an alliance.  Quotas
 $(0.9,0.9,2.2)$ leave two seats to be awarded at random.  In this instance,
 quota and ex ante proportionality already determine the relevant
 probabilities, independently of the method: exactly two parties are rounded
 up, with marginal probabilities $0.9,0.9,0.2$.  The alliance therefore
 receives two seats exactly when $C$ is not rounded up, an event of probability
 $0.8$.

Suppose the quotas change to $(1.03,0.92,2.05)$: both allies gain and their
opponent loses.  Party $C$ now receives a third seat with probability $0.05$,
so the alliance receives at least two seats with probability $0.95$.  Its seat
distribution has improved.  Nevertheless, the probability that $A$ and $B$
are both \emph{rounded up} falls from $0.8$ to zero, because $A$'s formerly
random seat has become part of its lower quota.  Comparing seat thresholds
correctly treats guaranteed and randomly awarded seats alike. Figure~\ref{fig:example} displays both comparisons.
\end{example}

\begin{figure}[H]
\centering
\begin{tikzpicture}
\begin{groupplot}[
  group style={group size=2 by 1, horizontal sep=1.6cm},
  paperplot,
  height=4.6cm,
  ybar, /pgf/bar width=9.5pt,
  ymin=0, ymax=1.15,
  ytick={0,0.2,...,1},
  enlarge x limits=0.35,
  axis x line*=bottom, axis y line*=left,
  axis line style={black!45},
  tick style={black!45},
  ymajorgrids, grid style={black!12},
  nodes near coords, nodes near coords style={font=\scriptsize},
  point meta=rawy,
]
\nextgroupplot[
  width=6.6cm,
  xtick={1,2,3}, xticklabels={$k=1$,$k=2$,$k=3$},
  ylabel={$\Pr[X_{\{A,B\}}\ge k]$},
  legend to name=leg:example, legend columns=2,
]
\addplot[fill=ravgray!28, draw=ravgray!75, area legend,
  nodes near coords style={text=ravgray}] coordinates {(1,1) (2,0.80) (3,0)};
\addplot[fill=ravblue!88, draw=ravblue, area legend,
  nodes near coords style={text=ravblue}] coordinates {(1,1) (2,0.95) (3,0)};
\legend{{$q=(0.9,0.9,2.2)$},{$q'=(1.03,0.92,2.05)$}}
\nextgroupplot[
  width=4.2cm,
  xtick={1}, xticklabels={{$\Pr[A\text{ and }B\text{ rounded up}]$}},
  ymin=0, ymax=1.15,
  xmin=0.45, xmax=1.55,
]
\addplot[fill=ravgray!28, draw=ravgray!75,
  nodes near coords style={text=ravgray}] coordinates {(1,0.80)};
\addplot[fill=ravblue!88, draw=ravblue,
  nodes near coords style={text=ravblue}] coordinates {(1,0)};
\end{groupplot}
\end{tikzpicture}\\[2pt]
\pgfplotslegendfromname{leg:example}
\caption{The alliance $\{A,B\}$ of Example~\ref{example:AB} with house size four.
Left: upper-tail probabilities of the alliance's seat total before and after
the change from $q$ to $q'$; every tail weakly increases, so the seat
distribution improves in stochastic dominance. Right: the probability that
both allies are rounded up falls from $0.8$ to $0$, because $A$'s formerly
random seat has become part of its lower quota. Monotonicity must therefore
be stated in terms of seat thresholds rather than rounding events.}
\label{fig:example}
\end{figure}
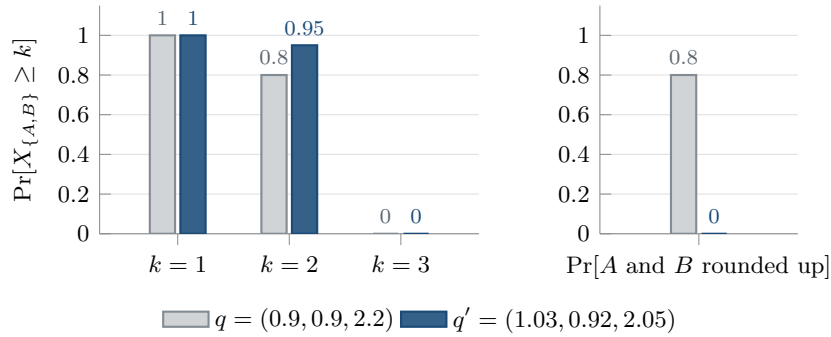

\citet{CGST+26a} describe threshold monotonicity as a general notion of
strategy-proofness for randomized apportionment: a voter cannot profit by
supporting a disapproved party, and the same protection extends to coalitions
of voters.
\citet{CGST+24a,CGST+26a} conjectured that Sampford's quota-compliant, ex-ante
proportional method satisfies threshold monotonicity for arbitrary coalitions.
The conjecture also appears in a 2025 survey~\citep{Aziz25a}.

Sampford sampling \citep{Samp67a} is a classical unequal-probability sampling
scheme: here it selects the parties to be rounded up with inclusion
probabilities equal to their quota residues. Threshold monotonicity is stronger
than the selection monotonicity studied by \citet{CGST+26a}, who proved that
Sampford's method satisfies selection monotonicity. Threshold monotonicity
implies this property. The empty-sample case is immediate. Otherwise, use
$q=p$, $q'=p'$, and $h=\sum_i p_i=\sum_i p'_i$. A feasible set is selected
exactly when its $h$ members receive $h$ seats in total. Selection monotonicity
therefore compares the probability of a specified feasible set when its
members' residues increase and all other residues decrease.
It controls one specified set of rounded-up parties. However, a
threshold event is a union of many such outcomes.  After a coalition is
reinforced, some outcomes in this union can become less likely because they
also round up parties outside the coalition.  The selection-monotonicity
inequalities therefore cannot simply be added: a proof for threshold monotonicity must show that losses in
some outcomes are compensated by gains in others.

\subsection{Contributions}

We prove the main conjecture of \citet{CGST+24a,CGST+26a}, giving
the first established compatibility of quota, ex-ante proportionality, and
threshold monotonicity.  
The proof moves along the straight segment from the original quota vector to the reinforced one. Between integer crossings, we condition on whether the varied party is selected. The only possibly negative contribution to the derivative of a threshold probability is then a single covariance. Newton's inequalities confine that contribution to outcomes falling exactly one coalition member short of the threshold. A coefficient bound and double count show that the positive term in the derivative is at least as large. Along the segment, coalition residues rise and outside residues fall, so every threshold probability is nondecreasing between crossings. At a crossing, a seat merely changes its description, from residual to guaranteed or conversely, without changing the complete allocation law. Continuity there joins the pieces.


We also prove a complementary impossibility for the pairwise threshold monotonicity axiom introduced
by \citet{CGST+26a}. This axiom compares a coalition whose members gain with a
second coalition whose members lose, leaves all other parties unrestricted,
and requires at least one of the two coalition seat distributions to move in
the expected direction. Even when the tested coalitions are required to be
disjoint, quota and ex-ante proportionality are incompatible with this axiom
for every $n\geq7$. Unlike the impossibility theorem of \citet{CGST+26a}, our
argument does not assume full support. Our impossibility result can be viewed as a randomized and coalitional analogue of the well-known Balinski--Young impossibility theorem, which says that no quota-compliant deterministic method satisfies population monotonicity~\citep{BaYo01a}. The proof is graph-theoretic: proportionality turns the two-seat lottery into a fractional matching with prescribed vertex degrees, and, at the constructed profile, pairwise threshold monotonicity forces the positive edges between the small parties to be pairwise intersecting.

\section{Related work}
\label{sec:related-work}



\paragraph{Apportionment.}
Classical apportionment theory balances quota, population monotonicity, and
house monotonicity \citep{BaYo01a,Puke14a,RoUl10a,Szpi10a}.
House monotonicity requires that an increase in the total number of seats
should not hurt any party.  Population monotonicity, as \citet{Youn04a}
succinctly puts it, requires that ``a growing state should not give up seats
to a shrinking state''; its violation is the population
paradox~\citep[pp.~42--43]{BaYo01a}.  The Balinski--Young impossibility
theorem shows that the deterministic quota requirement is incompatible with
population monotonicity~\citep{BaYo01a,GPP25a}. More basically, seat
indivisibility prevents any deterministic method from giving every party its
exact fractional entitlement. Randomization changes the comparison. 

\paragraph{Randomized apportionment.}

 Although there is relatively less work on randomized apportionment, even countries committed to deterministic apportionment rely on chance at the margins. 
  Almost every divisor or quota method needs a tie-break rule, and a common one is the random draw which is used in some countries including Germany~\citep[Chapter 4]{Puke14a}.
 \citet{Grim04a} combines ex-post quota
with exact proportionality in expectation; \citet{GPP25a} add ex-post house
monotonicity through cumulative rounding; and \citet{CCS+25a} characterize the
randomized methods satisfying quota, ex-ante proportionality, and house
monotonicity via network-flow formulations of quota-compliant, house-monotone
paths and lotteries over them.  Proportional lotteries also arise in peer
selection and random allocation \citep{Aziz25a,ALM+19a,BCKM13a}.  These
guarantees either vary the house size or concern a single vote profile.
\citet{CGST+26a} instead keep the house fixed and ask how the \emph{joint}
distribution of rounded-up parties responds when a coalition is reinforced.
Their threshold monotonicity can be viewed as a coalition-level generalization
of population monotonicity for randomized methods.  They proved selection
monotonicity for Sampford rounding and conjectured the arbitrary-coalition
threshold result proved here.

\paragraph{Unequal-probability sampling.}
After lower quotas have been awarded, the residual allocation is exactly a
fixed-size unequal-probability ($\pi$ps) sampling problem: construct a
distribution over $k$-element samples whose first-order inclusion
probabilities equal prescribed values $\pi_i\in[0,1]$ with $\sum_i\pi_i=k$,
here with $\pi_i=p_i$ \citep{BrHa83a}.  
Sampford's design \citep{Samp67a} is one of these methods. 
The selection and threshold axioms of \citet{CGST+26a} draw
on systematic and unequal-probability sampling \citep{Mado49a,HoTh52a,HLS89a},
maximum-entropy designs \citep{CDL94a}, and the pivotal method
\citep{DeTi98a,Till06a,Till23a}.  Whereas the sampling problem starts from one
list of prescribed individual probabilities, the present question compares
every coalition threshold across two such profiles.

\paragraph{Dependent rounding.}
In computer science, related procedures are studied under the name dependent
rounding and use negative dependence for concentration and approximation guarantees
\citep{AgSv04a,BrJo12a,CJMW20a,GKPS06a}.  These describe one lottery
at one profile, whereas threshold monotonicity compares two.  

\paragraph{Fair Division.}
Apportionment can also be viewed as a special case of fair allocation with entitlements where the items 
are identical and homogeneous. For more on weighted fair division, see the survey of \citet{Suks25a}.
In that respect, randomized apportionment can be viewed as best-of-both-worlds fair division~\citep{AFS+23a} where we aim for fairness guarantees ex-ante and ex-post.

\section{Preliminaries}
\label{sec:preliminaries}

Let $N=[n]=\{1,\ldots,n\}$ be the set of parties. For a finite set $G$ and an
integer $r$, write $\binom{G}{r}=\{D\subseteq G:|D|=r\}$.

\subsection{From quota residues to a fixed-size lottery}

Let $v=(v_1,\ldots,v_n)\in\R_{\geq0}^n$ be a vote vector with
$\sum_{i\in N}v_i>0$, and let $h\in\N$ be the house size. An
\emph{apportionment} is a vector $x\in\Nzero^n$ satisfying
$\sum_{i\in N}x_i=h$. An \emph{apportionment method} $\mathcal A$ maps every
pair $(v,h)$ to a random apportionment $\mathcal A(v,h)$.

Party $i$'s \emph{standard quota}, \emph{lower quota}, and \emph{residue} are,
respectively,
\begin{equation}
q_i=\frac{h v_i}{\sum_{j\in N}v_j},
\qquad
b_i=\lfloor q_i\rfloor,
\qquad
p_i=q_i-b_i.
\label{eq:quota-residue}
\end{equation}
The number of seats remaining after all lower quotas have been awarded is
\begin{equation}
\alpha=h-\sum_{i\in N}b_i=\sum_{i\in N}p_i.
\label{eq:alpha}
\end{equation}
Thus $p\in[0,1)^n$ and $\alpha\in\{0,\ldots,n-1\}$.

An apportionment method satisfies \emph{quota} if every party receives either
its lower or upper quota in every realization.  It is \emph{ex ante
proportional} if $\E[\mathcal A(v,h)_i]=q_i$ for every party $i$.

Following \citet{CGST+26a}, a \emph{selection lottery} $r$ takes as input a
vector $p\in[0,1)^n$ whose coordinates sum to an integer $\alpha$ and returns
a random set
\[
r(p)\in\binom{N}{\alpha}
\quad\text{such that}\quad
\Pr[i\in r(p)]=p_i\quad(i\in N).
\]
Every selection lottery induces an apportionment method by awarding lower
quotas and then one additional seat to each selected party:
\begin{equation}
\mathcal A^r(v,h)_i=b_i+\mathbf 1\{i\in r(p)\}.
\label{eq:induced-method}
\end{equation}
The induced method satisfies quota and ex ante proportionality. Selection
lotteries are also studied as unequal-probability sampling without replacement
and as dependent randomized rounding
\citep{BrHa83a,DeTi98a,GKPS06a}.

\subsection{Sampford rounding and its prescribed marginals}

For a residue vector $p\in [0,1)^n$, let
\[
 N'=\{i\in N:0<p_i<1\}.
\]
That is, $N'$ is  the set of parties
with fractional quotas.  The marginal condition in the definition of a
selection lottery already implies that parties outside $N'$ are never
selected.

\begin{remark}[Why not simply condition independent rounding?]
\label{rem:why-odds}
Selecting each party in $N'$ independently with probability $p_i$ gives the
desired individual probabilities, but the number selected is random.
Conditioning on selecting exactly $\alpha$ parties fixes the size.  For
$S\in\binom{N'}{\alpha}$,
\[
 \prod_{i\in S}p_i\prod_{j\in N'\setminus S}(1-p_j)
 =
 \left(\prod_{j\in N'}(1-p_j)\right)
 \prod_{i\in S}\frac{p_i}{1-p_i}.
\]
The first factor is common to every $S$, so the conditioned sampling weights $S$
by the product of the odds $w_i=p_i/(1-p_i)$.  This explains the odds in
Sampford's rule, but it does not recover the desired marginals: conditioning
generally changes each party's selection probability.  The odds merely
rewrite the original independent-rounding  after conditioning; they are not
recalibrated probabilities.  Conditional Poisson sampling can instead search
for different starting probabilities whose conditional marginals are $p_i$
\citep{Haje81a,CDL94a,Till06a}.  Sampford instead corrects the conditioned set
weights themselves.
\end{remark}

\begin{definition}[Sampford rounding]
\label{def:sampford}
For $\alpha=0$, set $r^{\Samp}(p)=\varnothing$.  For $\alpha\geq1$, all draws are
from $N'$.  Choose the distinguished party by giving each $i\in N'$ probability
$p_i/\alpha$.  Fill the other $\alpha-1$ slots independently with replacement,
giving each $i\in N'$ probability $w_i/\sum_{j\in N'}w_j$ on every draw, where
$w_i=p_i/(1-p_i)$.  
\harisnew{If all $\alpha$
parties are distinct and form some set $S$, select $S$ and set
$r^{\Samp}(p)=S$; otherwise restart.}
The resulting
selection lottery is \emph{Sampford rounding}, and the method induced through
\eqref{eq:induced-method} is the \emph{Sampford method}.  Parties outside $N'$
have zero residue and are never selected.
\end{definition}

The distinguished draw is simply one differently weighted slot.  Calling it
the ``first'' draw describes a convenient implementation; it is not a
preliminary reweighting phase and could just as well be sampled last.
Polynomial-time implementations that avoid restarts are also available
\citep{Graf09a}.
Assume for the rest of this derivation that $\alpha\geq1$.
For any $D\subseteq N'$, define its \emph{correction factor}
\begin{equation}
c_p(D)=\sum_{i\in D}(1-p_i).
\label{eq:correction-factor}
\end{equation}
If $S\in\binom{N'}{\alpha}$, then
\[
 c_p(S)=\sum_{j\in N'\setminus S}p_j,
\]
because $|S|=\sum_{i\in N'} p_i=\alpha$.  Thus $c_p(S)$ is both the total amount by
which the selected parties are rounded up and the fractional mass left outside
$S$.

The distinguished slot produces the following correction.  Fix a
possible output $S$.  If $i\in S$ occupies that slot, its contribution, apart
from factors common to every output, is
\[
p_i\prod_{j\in S\setminus\{i\}}w_j
=(1-p_i)\prod_{j\in S}w_j.
\]
The remaining elements can occur in any order, but the number of orderings is
the same for every $S$ and cancels when we condition on distinct draws.
Summing over the possible distinguished party gives
\[
\sum_{i\in S}p_i\prod_{j\in S\setminus\{i\}}w_j
=c_p(S)\prod_{j\in S}w_j.
\]
Thus the distinguished draw multiplies the conditioned independent-rounding
weight by exactly the correction factor $c_p(S)$. The restart construction
therefore gives, for $S\in\binom{N'}{\alpha}$,
\begin{equation}
 \Pr_p[r^{\Samp}(p)=S]
 =\frac{
  c_p(S)\displaystyle\prod_{i\in S}p_i
       \displaystyle\prod_{j\in N'\setminus S}(1-p_j)}{
  \displaystyle\sum_{R\in\binom{N'}{\alpha}}
  c_p(R)\prod_{i\in R}p_i
       \prod_{j\in N'\setminus R}(1-p_j)}.
\label{eq:sampford-pmf}
\end{equation}

For $\alpha\geq1$, this formula covers every residue vector in $[0,1)^n$:
zero-residue parties have already been omitted.  The case $\alpha=0$ was
defined separately.  One-sided limits in which a residue approaches one are
treated in the proof of Theorem~\ref{thm:sampford-threshold}; see also
\citet[Equation~(2)]{CGST+26a}.

The correction factor restores the prescribed marginals. To see this, fix a
party $k\in N'$ and use the equivalent raw weights
$c_p(S)\prod_{i\in S}w_i$.  Dividing the total weight of sets containing $k$
by $w_k$ gives
\begin{align*}
&\frac{1}{w_k}
 \sum_{\substack{S\in\binom{N'}{\alpha}\\k\in S}}
 c_p(S)\prod_{i\in S}w_i\\
&\quad=
 \sum_{D\in\binom{N'\setminus\{k\}}{\alpha-1}}
 \left(\sum_{\ell\in N'\setminus(D\cup\{k\})}p_\ell\right)
 \prod_{i\in D}w_i\\
&\quad=
 \sum_{S\in\binom{N'\setminus\{k\}}{\alpha}}
 c_p(S)\prod_{i\in S}w_i.
\end{align*}
The middle equality pairs $D$ with a party
$\ell\in N'\setminus(D\cup\{k\})$, puts $S=D\cup\{\ell\}$, and uses
$p_\ell=(1-p_\ell)w_\ell$.  The final expression is the total weight of sets
that omit $k$.  Sets containing $k$ therefore have $w_k$ times as much total
weight as sets omitting it, and hence
\[
\Pr[k\in r^{\Samp}(p)]
=\frac{w_k}{1+w_k}
=p_k.
\]
Thus the correction factor created by the distinguished draw exactly repairs
the marginal distortion caused by conditioning.  Parties outside $N'$ have
both selection probability and prescribed marginal equal to zero, so Sampford
rounding has the prescribed inclusion probabilities \citep{Samp67a}.

For example, let $p=(0.9,0.9,0.2)$ and $\alpha=2$.  Conditioning independent
rounding gives the pairs $\{A,B\},\{A,C\},\{B,C\}$ odds-product weights
$(81,2.25,2.25)$, hence probabilities
$(18/19,1/38,1/38)$ and marginals $(37/38,37/38,1/19)$---not
$(0.9,0.9,0.2)$. The three correction factors are $(0.2,0.9,0.9)$.
Multiplying by them and normalizing gives Sampford probabilities
$(0.8,0.1,0.1)$, whose marginals are exactly $(0.9,0.9,0.2)$.

\subsection{Coalition reinforcement and stochastic dominance}
\begin{definition}[Stochastic dominance]
For nonnegative integer-valued random variables $X$ and $Y$, write
$X\succeq_{\SD}Y$ if
\[
\Pr[X\geq t]\geq\Pr[Y\geq t]
\quad\text{for every }t\in\Nzero.
\]
Equivalently, $\E[u(X)]\geq\E[u(Y)]$ for every nondecreasing function $u$ for
which the expectations exist. This is also called first-order stochastic
dominance.
\end{definition}

For a coalition $T\subseteq N$, define its random seat total under method
$\mathcal A$ by
\begin{equation}
X_T(v,h)=\sum_{i\in T}\mathcal A(v,h)_i.
\label{eq:coalition-total}
\end{equation}

\begin{definition}[Threshold monotonicity \citep{CGST+26a}]
\label{def:threshold-monotonicity}
Let $v,v'\in\R_{\geq0}^n$ be two vote vectors, each with positive total, let
$h\in\N$, and let $q,q'$ be the corresponding standard-quota vectors. An
apportionment method $\mathcal A$ satisfies \emph{threshold monotonicity} if,
for every coalition $T\subseteq N$,
\begin{equation}
q'_i\geq q_i\quad(i\in T),
\qquad
q'_j\leq q_j\quad(j\notin T)
\label{eq:coalition-reinforcement}
\end{equation}
implies
\[
X_T(v',h)\succeq_{\SD}X_T(v,h).
\]
\end{definition}

The Sampford method depends on $(v,h)$ only through the standard-quota vector.
Accordingly, we also write $X_T(q)$ for the random variable in
\eqref{eq:coalition-total}.

Under Sampford rounding,
\[
X_T(q)
=
\sum_{i\in T}\lfloor q_i\rfloor
+\bigl|r^{\Samp}(p)\cap T\bigr|.
\]
Thus, while the lower quotas are fixed, the only random part of the
coalition's seat total is its number of members in the residual sample.

\section{Threshold monotonicity of the Sampford method}
\label{sec:threshold-sampford}

\harisnew{In this section, we prove that the Sampford method satisfies threshold monotonicity.}

\begin{theorem}[Sampford threshold monotonicity]
\label{thm:sampford-threshold}
The Sampford method satisfies threshold monotonicity.
\end{theorem}

Between integer crossings, lower quotas are fixed.  We will show that the
partial derivative of an upper-tail probability with respect to a party's
residue is nonnegative for a party in the coalition and nonpositive for a
party outside it.  Along the straight reinforcement path, coalition residues
rise and outside residues fall, so these signs make every upper-tail
probability nondecreasing.

\subsection{Adjacent product-weight sample sizes}

Although the Sampford method itself includes the correction factor \(c_p(S)\), our proof uses auxiliary fixed-size samples whose probabilities are proportional only to the product of their odds weights. We need the following elementary fact about such samples.

\begin{lemma}
\label{lem:adjacent-product-counts}
Let $G$ be a finite set with positive weights.  For $0\leq s\leq|G|$, let
$H_s$ be an $s$-subset drawn with probability proportional to the product of
its members' weights; in particular, $H_0=\varnothing$.  Then, for every
$U\subseteq G$ and $1\leq s\leq|G|$,
\[
 |H_s\cap U|
 \ \preceq_{\SD}\
 |H_{s-1}\cap U|+1.
\]
\end{lemma}

\begin{proof}
If $U$ is empty, the conclusion is immediate.  Assume from now on that $U$ is
nonempty.
For $A\subseteq G$, let $e_\ell(A)$ be the sum of the products of the weights
over all $\ell$-subsets of $A$, with $e_0(A)=1$.  Then
\[
 \Pr[|H_s\cap U|=\ell]
 =\frac{e_\ell(U)e_{s-\ell}(G\setminus U)}{e_s(G)}.
\]
On the common support of $|H_s\cap U|-1$ and $|H_{s-1}\cap U|$,
\[
 \frac{\Pr[|H_s\cap U|-1=\ell]}
      {\Pr[|H_{s-1}\cap U|=\ell]}
 =
 \frac{e_{s-1}(G)}{e_s(G)}
 \frac{e_{\ell+1}(U)}{e_\ell(U)}.
\]
For $1\leq\ell\leq |U|-1$, Newton's inequalities give
\[
 e_\ell(U)^2\geq e_{\ell-1}(U)e_{\ell+1}(U),
\]
so $e_{\ell+1}(U)/e_\ell(U)$ is nonincreasing wherever the ratio is defined.
The two probability masses therefore cross at most once: the first may exceed
the second only at lower values.  Any mass outside their common support lies
below it for $|H_s\cap U|-1$ and above it for $|H_{s-1}\cap U|$.  Since both
distributions have total mass one, every upper tail of the first is no larger
than the corresponding upper tail of the second.  This proves the lemma.
\end{proof}

\subsection{How one residue changes a threshold probability}

We first work with a fixed set of fractional parties and at least two residual
seats.  The cases with at most one residual seat and the integer crossings
themselves are handled in the final subsection.

For the partial derivatives below, extend the Sampford probabilities off the
hyperplane where the residues sum to $\alpha$ by assigning each
$S\in\binom{N'}{\alpha}$ probability proportional to
$c_p(S)\prod_{x\in S}p_x/(1-p_x)$.  This gives a smooth distribution for every
$p\in(0,1)^{N'}$ and agrees with Sampford rounding whenever
$\sum_{x\in N'}p_x=\alpha$.

\begin{lemma}
\label{lem:coordinate-tail-signs}
Let $N'\subseteq N$, let $p\in(0,1)^{N'}$ satisfy
\[
 \sum_{x\in N'}p_x=\alpha\geq2
\]
for an integer $\alpha$, and let $T\subseteq N$.  For every integer $k\geq0$,
\[
 \frac{\partial}{\partial p_i}
 \Pr_p[|r^{\Samp}(p)\cap T|\geq k]\geq0
 \quad\text{for }i\in T\cap N',
\]
whereas
\[
 \frac{\partial}{\partial p_j}
 \Pr_p[|r^{\Samp}(p)\cap T|\geq k]\leq0
 \quad\text{for }j\in N'\setminus T.
\]
\end{lemma}

\begin{proof}
All derivatives are evaluated at the feasible vector.  We use
$\sum_{x\in N'}p_x=\alpha$ only after differentiation.

It is enough to prove the first inequality for a nonconstant threshold.  Fix
$i\in T\cap N'$, remove $i$, and draw an $(\alpha-1)$-subset $H$ of
$G=N'\setminus\{i\}$ with probability proportional to the product of its
members' odds.  Write $Y=|H\cap T|$ for its coalition count.

\paragraph{Separating outcomes according to whether $i$ is selected.}
Divide every aggregate raw weight below by the product-weight normalizer used
for $H$.  If $p_i$ is varied to $t$ while the other residues are held fixed,
then, for a possible value $D$ of $H$, the raw weight of $D\cup\{i\}$ is
\[
 \left(\frac{t}{1-t}c_p(D)+t\right)
 \prod_{x\in D}w_x.
\]
Consequently the aggregate weights of all outcomes containing $i$ and of the
threshold outcomes containing $i$ are, respectively,
\begin{equation}
 \frac{t}{1-t}\E[c_p(H)]+t
 \quad\text{and}\quad
 \frac{t}{1-t}
 \E\!\left[c_p(H)\mathbf1\{Y\geq k-1\}\right]
 +t\Pr[Y\geq k-1].
 \label{eq:containing-raw-weights}
\end{equation}

The weights of outcomes omitting $i$ do not depend on $t$.  Their values at
the feasible vector follow from
\begin{equation}
 c_p(S)\prod_{x\in S}w_x
 =\sum_{x\in S}p_x\prod_{y\in S\setminus\{x\}}w_y.
 \label{eq:insertion-identity}
\end{equation}
After writing $H=S\setminus\{x\}$, the aggregate weight of all such outcomes
is
\[
 \E\!\left[\sum_{x\in G\setminus H}p_x\right]
 =\E[c_p(H)]+1-p_i,
\]
where the equality uses feasibility.  For the threshold outcomes, the same
reindexing gives
\begin{align}
 &\E\!\left[
  \sum_{x\in G\setminus H}p_x
  \mathbf1\{Y+\mathbf1\{x\in T\}\geq k\}
 \right]
 \notag\\
 &\quad=
 \E\!\left[c_p(H)\mathbf1\{Y\geq k-1\}\right]
 +(1-p_i)\Pr[Y\geq k-1]
 -\E\!\left[
  \mathbf1\{Y=k-1\}
  \sum_{\substack{x\in G\setminus H\\x\notin T}}p_x
 \right].
 \label{eq:omitting-threshold-weight}
\end{align}
The final term removes precisely the outside completions of a partial set
that is one coalition member short.

The weights of outcomes omitting $i$ contain no $t$.  Feasibility was used
only to rewrite their constant values at the base vector in the forms above.
Adding the containing and omitting weights and then setting $t=p_i$ shows that
the total raw weight and its derivative are
\begin{equation}
 \frac{\E[c_p(H)]+1-p_i}{1-p_i}
 \quad\text{and}\quad
 \frac{\E[c_p(H)]}{(1-p_i)^2}+1,
 \label{eq:total-weight-and-derivative}
\end{equation}
whereas the threshold raw weight and its derivative are
\begin{align}
 &\frac{
   \E\!\left[c_p(H)\mathbf1\{Y\geq k-1\}\right]}
  {1-p_i}
 +\Pr[Y\geq k-1]
 -\E\!\left[
  \mathbf1\{Y=k-1\}
  \sum_{\substack{x\in G\setminus H\\x\notin T}}p_x
 \right],
 \notag\\[-0.25ex]
 &\hspace{28mm}\text{and}\qquad
 \frac{
   \E\!\left[c_p(H)\mathbf1\{Y\geq k-1\}\right]}
  {(1-p_i)^2}
 +\Pr[Y\geq k-1].
 \label{eq:threshold-weight-and-derivative}
\end{align}
The quotient rule now gives the exact derivative
\begin{align}
 &\frac{\partial}{\partial p_i}
 \Pr_p[|r^{\Samp}(p)\cap T|\geq k]=
 \notag\\
 &=
 \frac{
  \bigl(\E[c_p(H)]+(1-p_i)^2\bigr)
  \E\!\left[
   \mathbf1\{Y=k-1\}
   \sum_{\substack{x\in G\setminus H\\x\notin T}}p_x
  \right]
  +p_i\operatorname{Cov}\!\left(
   c_p(H),\mathbf1\{Y\geq k-1\}
  \right)
 }{\bigl(\E[c_p(H)]+1-p_i\bigr)^2}
 .
 \label{eq:threshold-derivative}
\end{align}
Only the covariance in this formula can be negative.

\paragraph{Accounting for the correction factor.}
Expanding $c_p(H)=\sum_{x\in H}(1-p_x)$ gives
\begin{align}
 &\operatorname{Cov}\!\left(c_p(H),\mathbf1\{Y\geq k-1\}\right)
 =
 \sum_{x\in G}(1-p_x)\Pr[x\in H]
 \left(
  \Pr[Y\geq k-1\mid x\in H]-\Pr[Y\geq k-1]
 \right).
 \label{eq:covariance-expansion}
\end{align}
Conditional on $x\notin H$, the set $H$ is a product-weighted
$(\alpha-1)$-set on $G\setminus\{x\}$. Conditional on $x\in H$, the set
$H\setminus\{x\}$ is a product-weighted $(\alpha-2)$-set on the same ground
set. These are the two adjacent sizes in
Lemma~\ref{lem:adjacent-product-counts}.  When $x\in T$, the latter sample also
comes with the fixed coalition member $x$, so its summand in
\eqref{eq:covariance-expansion} is nonnegative.  When $x\notin T$, the same
lemma says
\[
 \Pr[Y\geq k-1\mid x\in H]
 \geq
 \Pr[Y\geq k\mid x\notin H].
\]
Therefore
\begin{align}
 &\Pr[Y\geq k-1]-\Pr[Y\geq k-1\mid x\in H]
 \notag\\
 &\quad=
 \Pr[x\notin H]
 \left(
  \Pr[Y\geq k-1\mid x\notin H]
  -\Pr[Y\geq k-1\mid x\in H]
 \right)
 \notag\\
 &\quad\leq \Pr[Y=k-1,\ x\notin H].
 \label{eq:outsider-boundary-loss}
\end{align}
Thus an outside party can contribute negatively only when the partial sample
has exactly $k-1$ coalition members.

We also need the following coefficient bound, which follows from feasibility.
For every $x\in G$,
\begin{align}
 \E[c_p(H)]
 &=
 \frac{
  \displaystyle
  \sum_{J\in\binom{G}{\alpha-2}}
  \left(\prod_{z\in J}w_z\right)
  \sum_{r\in G\setminus J}p_r}
 {\displaystyle
  \sum_{D\in\binom{G}{\alpha-1}}\prod_{z\in D}w_z}
\geq
 \frac{
  \displaystyle
  \sum_{J\in\binom{G\setminus\{x\}}{\alpha-2}}
  \prod_{z\in J}w_z}
 {\displaystyle
  \sum_{D\in\binom{G}{\alpha-1}}\prod_{z\in D}w_z}
 =\frac{1-p_x}{p_x}\Pr[x\in H].
 \label{eq:coefficient-bound}
\end{align}
The first equality expands $c_p(H)$ party by party, removes that party from
$H$, and uses $(1-p_r)w_r=p_r$. For every
$J\in\binom{G}{\alpha-2}$, we have
\[
 \sum_{r\in G\setminus J}p_r=c_p(J)+2-p_i\geq1,
\]
which proves the inequality.  Equivalently,
\[
 (1-p_x)\Pr[x\in H]\leq p_x\E[c_p(H)].
\]
Thus, for each outside party, the magnitude of its negative summand in
\eqref{eq:covariance-expansion} is at most
\[
 \E[c_p(H)]p_x\Pr[Y=k-1,\ x\notin H].
\]
Summing this bound over the outside parties gives
\begin{align}
 &\operatorname{Cov}\!\left(c_p(H),\mathbf1\{Y\geq k-1\}\right)
 \geq
 -\E[c_p(H)]
 \E\!\left[
  \mathbf1\{Y=k-1\}
  \sum_{\substack{x\in G\setminus H\\x\notin T}}p_x
 \right].
 \label{eq:covariance-bound}
\end{align}

The expectation in \eqref{eq:covariance-bound} is the same boundary term that
appears positively in \eqref{eq:threshold-derivative}.  Substitution leaves
\begin{align}
 &\frac{\partial}{\partial p_i}
 \Pr_p[|r^{\Samp}(p)\cap T|\geq k]
 \geq
 \frac{1-p_i}{\E[c_p(H)]+1-p_i}
 \E\!\left[
  \mathbf1\{Y=k-1\}
  \sum_{\substack{x\in G\setminus H\\x\notin T}}p_x
 \right]
 \geq0.
 \label{eq:insider-derivative-bound}
\end{align}

Finally, let $j\in N'\setminus T$.  If $k>\alpha$, the event in the lemma is
impossible, so its derivative is zero.  If $k\leq\alpha$, apply the inequality
just proved to the complementary coalition.  Since every residual sample has
size $\alpha$,
\[
 \Pr_p[|r^{\Samp}(p)\cap T|\geq k]
 =1-
 \Pr_p[|r^{\Samp}(p)\cap(N'\setminus T)|\geq\alpha-k+1].
\]
Increasing $p_j$ makes the probability after the minus sign nondecreasing,
so the derivative of the original upper tail is nonpositive.
\end{proof}

\subsection{Proof of Theorem~\ref{thm:sampford-threshold}}

Figure~\ref{fig:path} illustrates the interval decomposition used below and
the corresponding coalition upper-tail probabilities along one representative
path.

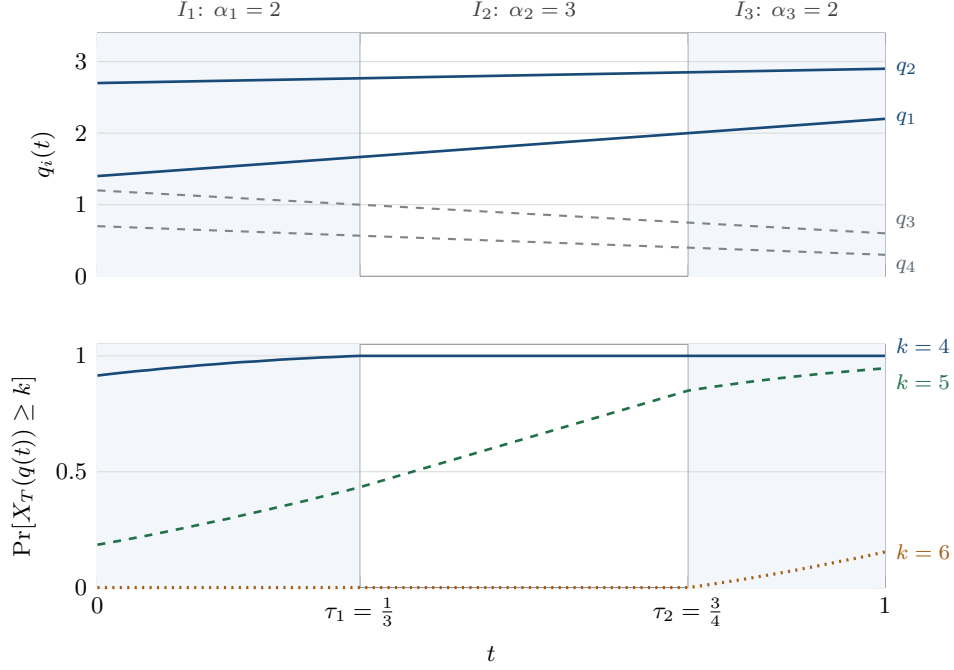
\begin{figure}[!t]
\centering
\begin{tikzpicture}
\begin{groupplot}[
  group style={group size=1 by 2, vertical sep=0.9cm, x descriptions at=edge bottom},
  paperplot,
  width=12cm,
  xmin=0, xmax=1,
  xtick={0,0.3333,0.75,1}, xticklabels={$0$,$\tau_1=\tfrac13$,$\tau_2=\tfrac34$,$1$},
  axis line style={black!45},
  tick style={black!45},
  clip=false,
]
\nextgroupplot[
  height=4.8cm,
  ymin=0, ymax=3.4, ytick={0,1,2,3},
  ylabel={$q_i(t)$},
  xlabel={},
]
\fill[ravlight] (0,0) rectangle (0.3333,3.4);
\fill[ravlight] (0.75,0) rectangle (1,3.4);
\draw[black!14] (0,1) -- (1,1);
\draw[black!14] (0,2) -- (1,2);
\draw[black!14] (0,3) -- (1,3);
\draw[ravgray!50] (0.3333,0) -- (0.3333,3.4);
\draw[ravgray!50] (0.75,0)   -- (0.75,3.4);
\addplot[ravblue, line width=1pt] coordinates {(0,1.4) (1,2.2)} node[right,font=\scriptsize,text=ravblue] {$q_1$};
\addplot[ravblue, line width=1pt] coordinates {(0,2.7) (1,2.9)} node[right,font=\scriptsize,text=ravblue] {$q_2$};
\addplot[ravgray!80, dashed] coordinates {(0,1.2) (1,0.6)} node[above right,font=\scriptsize,text=ravgray,yshift=-2pt] {$q_3$};
\addplot[ravgray!80, dashed] coordinates {(0,0.7) (1,0.3)} node[below right,font=\scriptsize,text=ravgray,yshift=2pt] {$q_4$};
\node[font=\scriptsize, text=black!75, anchor=south] at (0.1667,3.45) {$I_1$: $\alpha_1=2$};
\node[font=\scriptsize, text=black!75, anchor=south] at (0.5417,3.45) {$I_2$: $\alpha_2=3$};
\node[font=\scriptsize, text=black!75, anchor=south] at (0.875,3.45)  {$I_3$: $\alpha_3=2$};
\nextgroupplot[
  height=4.8cm,
  ymin=0, ymax=1.05, ytick={0,0.5,1},
  ylabel={$\Pr[X_T(q(t))\ge k]$},
  xlabel={$t$},
]
\fill[ravlight] (0,0) rectangle (0.3333,1.05);
\fill[ravlight] (0.75,0) rectangle (1,1.05);
\draw[black!14] (0,0.5) -- (1,0.5);
\draw[black!14] (0,1) -- (1,1);
\draw[ravgray!50] (0.3333,0) -- (0.3333,1.05);
\draw[ravgray!50] (0.75,0)   -- (0.75,1.05);
\addplot[ravblue, line width=1pt] coordinates {(0.000,0.915) (0.017,0.922) (0.033,0.928) (0.050,0.934) (0.067,0.939) (0.083,0.945) (0.100,0.950) (0.117,0.955) (0.133,0.960) (0.150,0.964) (0.167,0.969) (0.183,0.973) (0.200,0.976) (0.217,0.980) (0.233,0.984) (0.250,0.987) (0.267,0.990) (0.283,0.993) (0.300,0.995) (0.317,0.998) (0.333,1.000) (0.350,1.000) (0.367,1.000) (0.383,1.000) (0.400,1.000) (0.417,1.000) (0.433,1.000) (0.450,1.000) (0.467,1.000) (0.483,1.000) (0.500,1.000) (0.517,1.000) (0.533,1.000) (0.550,1.000) (0.567,1.000) (0.583,1.000) (0.600,1.000) (0.617,1.000) (0.633,1.000) (0.650,1.000) (0.667,1.000) (0.683,1.000) (0.700,1.000) (0.717,1.000) (0.733,1.000) (0.750,1.000) (0.767,1.000) (0.783,1.000) (0.800,1.000) (0.817,1.000) (0.833,1.000) (0.850,1.000) (0.867,1.000) (0.883,1.000) (0.900,1.000) (0.917,1.000) (0.933,1.000) (0.950,1.000) (0.967,1.000) (0.983,1.000) (1.000,1.000)} node[right,font=\scriptsize,text=ravblue,yshift=4pt] {$k=4$};
\addplot[ravgreen, dashed, line width=1pt] coordinates {(0.000,0.185) (0.017,0.195) (0.033,0.205) (0.050,0.216) (0.067,0.227) (0.083,0.238) (0.100,0.250) (0.117,0.262) (0.133,0.274) (0.150,0.286) (0.167,0.298) (0.183,0.311) (0.200,0.324) (0.217,0.337) (0.233,0.350) (0.250,0.363) (0.267,0.377) (0.283,0.391) (0.300,0.405) (0.317,0.419) (0.333,0.433) (0.350,0.450) (0.367,0.467) (0.383,0.483) (0.400,0.500) (0.417,0.517) (0.433,0.533) (0.450,0.550) (0.467,0.567) (0.483,0.583) (0.500,0.600) (0.517,0.617) (0.533,0.633) (0.550,0.650) (0.567,0.667) (0.583,0.683) (0.600,0.700) (0.617,0.717) (0.633,0.733) (0.650,0.750) (0.667,0.767) (0.683,0.783) (0.700,0.800) (0.717,0.817) (0.733,0.833) (0.750,0.850) (0.767,0.859) (0.783,0.867) (0.800,0.875) (0.817,0.883) (0.833,0.890) (0.850,0.897) (0.867,0.904) (0.883,0.910) (0.900,0.916) (0.917,0.922) (0.933,0.927) (0.950,0.932) (0.967,0.937) (0.983,0.942) (1.000,0.946)} node[right,font=\scriptsize,text=ravgreen,yshift=-5pt] {$k=5$};
\addplot[ravorange, dotted, line width=1.2pt] coordinates {(0.000,0.000) (0.017,0.000) (0.033,0.000) (0.050,0.000) (0.067,0.000) (0.083,0.000) (0.100,0.000) (0.117,0.000) (0.133,0.000) (0.150,0.000) (0.167,0.000) (0.183,0.000) (0.200,0.000) (0.217,0.000) (0.233,0.000) (0.250,0.000) (0.267,0.000) (0.283,0.000) (0.300,0.000) (0.317,0.000) (0.333,0.000) (0.350,0.000) (0.367,0.000) (0.383,0.000) (0.400,0.000) (0.417,0.000) (0.433,0.000) (0.450,0.000) (0.467,0.000) (0.483,0.000) (0.500,0.000) (0.517,0.000) (0.533,0.000) (0.550,0.000) (0.567,0.000) (0.583,0.000) (0.600,0.000) (0.617,0.000) (0.633,0.000) (0.650,0.000) (0.667,0.000) (0.683,0.000) (0.700,0.000) (0.717,0.000) (0.733,0.000) (0.750,0.000) (0.767,0.008) (0.783,0.016) (0.800,0.025) (0.817,0.034) (0.833,0.043) (0.850,0.053) (0.867,0.063) (0.883,0.073) (0.900,0.084) (0.917,0.095) (0.933,0.106) (0.950,0.118) (0.967,0.129) (0.983,0.141) (1.000,0.154)} node[right,font=\scriptsize,text=ravorange] {$k=6$};
\end{groupplot}
\end{tikzpicture}
\caption{The straight quota path $q(t)=(1-t)q+tq'$ for four parties,
$h=6$, coalition $T=\{1,2\}$, $q=(1.4,2.7,1.2,0.7)$ and
$q'=(2.2,2.9,0.6,0.3)$; coalition quotas are solid, outside quotas dashed. Top: the two interior breakpoints
$\tau_1$ and $\tau_2$ are the times at which a quota crosses an integer;
between them the lower quotas, the set of fractional parties and the number
$\alpha_\ell$ of residual seats are fixed. Bottom: the Sampford upper-tail
probabilities $F_k(t)=\Pr[X_T(q(t))\geq k]$, computed directly from
\eqref{eq:sampford-pmf}, for three thresholds. Each $F_k$ is nondecreasing on
every open interval by Lemma~\ref{lem:coordinate-tail-signs} and continuous at the breakpoints, which
together give $F_k(1)\ge F_k(0)$.}
\label{fig:path}
\end{figure}

We divide the straight quota path at its integer crossings.  On each resulting
open interval, the set of parties with fractional quotas is fixed, so the
derivative calculation applies there.  We then prove continuity of the
complete allocation at the crossings and join the interval comparisons.  No
derivative at a crossing is needed.

\begin{proof}[Proof of Theorem~\ref{thm:sampford-threshold}]
The statement is immediate when $T=\varnothing$ or $T=N$, so assume that both
$T$ and its complement are nonempty.  Let $q,q'$ satisfy
\eqref{eq:coalition-reinforcement}, set
\[
 q(t)=(1-t)q+tq'\qquad(0\leq t\leq1),
\]
and fix a complete-seat threshold $k\in\Nzero$.  Write
\[
 F_k(t)=\Pr[X_T(q(t))\geq k].
\]
We will prove that $F_k(1)\geq F_k(0)$.

Let
\[
 0=\tau_0<\tau_1<\cdots<\tau_M=1
\]
consist of the endpoints together with all distinct interior times at which a
nonconstant coordinate of $q(t)$ is an integer.  On
\[
 I_\ell=(\tau_{\ell-1},\tau_\ell),
\]
the lower quotas are fixed.  The set
\[
 N_\ell=\{i\in N:q_i(t)\notin\mathbb Z\},\qquad t\in I_\ell,
\]
is also fixed, and
\[
 p(t)=\bigl(q_i(t)-\lfloor q_i(t)\rfloor\bigr)_{i\in N_\ell}
 \in(0,1)^{N_\ell}.
\]
Parties outside $N_\ell$ have constant integer quotas on this interval and
receive only their deterministic lower-quota seats.  Moreover,
\[
 \alpha_\ell=\sum_{i\in N_\ell}p_i(t)
\]
is a fixed integer on $I_\ell$.

We first show that $F_k$ is nondecreasing on $I_\ell$.  The residual threshold
\[
 \kappa_\ell=k-\sum_{i\in T}\lfloor q_i(t)\rfloor
\]
is fixed there.  If $\kappa_\ell\leq0$ or $\kappa_\ell>\alpha_\ell$, the
corresponding upper tail is constant.  In particular, this covers every case
when $\alpha_\ell=0$.  If $\alpha_\ell=1$, the only nonconstant case is
$\kappa_\ell=1$, when
\[
 F_k(t)=\sum_{i\in T\cap N_\ell}p_i(t).
\]
This is nondecreasing because each quota in $T$ is nondecreasing.

Suppose now that $\alpha_\ell\geq2$.  The normalized raw-weight formula is
smooth throughout $(0,1)^{N_\ell}$.  The chain rule and
Lemma~\ref{lem:coordinate-tail-signs}, applied on the ground set $N_\ell$,
give
\begin{align*}
 F_k'(t)
 &=
 \sum_{i\in T\cap N_\ell}\dot p_i(t)
 \frac{\partial}{\partial p_i}
 \Pr[|r^{\Samp}(p(t))\cap T|\geq\kappa_\ell]
 +
 \sum_{j\in N_\ell\setminus T}\dot p_j(t)
 \frac{\partial}{\partial p_j}
 \Pr[|r^{\Samp}(p(t))\cap T|\geq\kappa_\ell]
 \geq0.
\end{align*}
Indeed, residues inside $T$ are nondecreasing, residues outside $T$ are
nonincreasing, and the two partial derivatives have the corresponding signs.
Thus $F_k$ is nondecreasing on every open interval $I_\ell$.

It remains to join these comparisons at the breakpoints.  Approach one such
breakpoint from an adjacent interval $I_\ell$, whose fixed set of fractional
parties is $N_\ell$.  If the interior residue vector converges to
$\bar p\in[0,1]^{N_\ell}$, then
\begin{equation}
 \begin{aligned}
 r^{\Samp}(p(t))
 &\xrightarrow{d}
 \{x\in N_\ell:\bar p_x=1\}\cup
 r^{\Samp}\bigl((\bar p_x)_{\,x\in N_\ell:\,0<\bar p_x<1}\bigr).
 \end{aligned}
 \label{eq:boundary-face-identity}
\end{equation}
Here $\bar p$ is only a one-sided limiting vector; an actual residue vector
still lies in $[0,1)^n$. 
If $\bar p_i=0$, then party $i$'s inclusion probability tends to zero, so
party $i$ is absent from the limiting sample.

When fractional coordinates remain in the limit,
\eqref{eq:boundary-face-identity} follows from
\eqref{eq:sampford-pmf}: every set of positive limiting weight contains all
unit coordinates and no zero coordinate, and its remaining factor is exactly
the Sampford weight on the fractional coordinates.  Their residues sum to
\[
 \alpha_\ell-\bigl|\{x\in N_\ell:\bar p_x=1\}\bigr|,
\]
an integer strictly between zero and their number, so the reduced Sampford sampling
has a positive normalizer.  If $\bar p$ is a zero--one vector, the prescribed
marginals give, for every interior feasible sequence $p^{(m)}\to\bar p$,
\[
 \Pr\!\left[
  r^{\Samp}(p^{(m)})\ne\{x\in N_\ell:\bar p_x=1\}
 \right]
 \leq
 \sum_{\substack{x\in N_\ell\\\bar p_x=1}}(1-p_x^{(m)})
 +
 \sum_{\substack{x\in N_\ell\\\bar p_x=0}}p_x^{(m)}
 \longrightarrow0.
\]

A quota whose residue tends to one approaches the breakpoint integer from
below.  Its lower quota on the interval is one less than at the breakpoint,
while its residual seat becomes certain.  A quota whose residue tends to zero
approaches the breakpoint integer from above, has the same lower quota as at
the breakpoint, and is omitted in the limit.
For the coordinates that remain fractional, the limiting distribution of the selected subset is precisely the distribution obtained by applying the Sampford method to their residues at the breakpoint. These limits hold jointly for all coordinates crossing an integer. At every interior breakpoint, the complete allocation law therefore converges from both sides to the law at the breakpoint; at each endpoint it converges from the one adjacent interval. Hence $F_k(t)$ is continuous on $[0,1]$.

Continuity extends each open-interval comparison to its endpoints.
Concatenating them yields $F_k(1)\geq F_k(0)$.  Since this holds for every
$k\in\Nzero$, we obtain
\[
 X_T(q')\succeq_{\SD}X_T(q),
\]
as required.
\end{proof}

The two extreme thresholds have useful interpretations for the residual
lottery.

\begin{corollary}[Extreme threshold events]
\label{cor:endpoint-events}
Let $p,p'\in[0,1)^n$ have the same integer coordinate sum, and let
$T\subseteq N$.  If
\[
 p'_i\geq p_i\quad(i\in T),
 \qquad
 p'_j\leq p_j\quad(j\notin T),
\]
then
\[
 \Pr[T\subseteq r^{\Samp}(p')]
 \geq
 \Pr[T\subseteq r^{\Samp}(p)],
 \qquad
 \Pr[r^{\Samp}(p')\subseteq T]
 \geq
 \Pr[r^{\Samp}(p)\subseteq T].
\]
\end{corollary}

\begin{proof}
If the common coordinate sum is zero, both selected sets are empty and the
claim is immediate.  Otherwise, let the common sum be $\alpha$ and apply
Theorem~\ref{thm:sampford-threshold} to the quota vectors $q=p$ and $q'=p'$.
They are realized by house size $h=\alpha$ and vote vectors $v=p$ and $v'=p'$.
The theorem therefore orders the two coalition counts in stochastic
dominance.  If $|T|\leq\alpha$, selecting every member of $T$ is their
upper-tail event at threshold $|T|$; if $|T|>\alpha$, that event is impossible.
If $|T|\geq\alpha$, selecting no outsider is their upper-tail event at
threshold $\alpha$; if $|T|<\alpha$, that event is impossible.
\end{proof}

The first inequality says that the probability of selecting every coalition
member cannot fall; the second says the same about selecting no outsider.
These are \textit{joint-inclusion monotonicity} and \textit{joint-exclusion monotonicity}, respectively.
When $|T|$ equals the residual sample size, both events say that the selected
set is exactly $T$.  Thus either inequality recovers Sampford's known
selection-monotonicity guarantee, while Theorem~\ref{thm:sampford-threshold}
also controls every intermediate coalition count.

\section{Limits of pairwise threshold monotonicity}
\label{sec:pairwise-impossibility}

Theorem~\ref{thm:sampford-threshold} concerns one reinforced coalition and
requires every party outside it to weaken. Pairwise threshold monotonicity
instead tests one gaining and one losing coalition while allowing all remaining
quotas to move in either direction; it requires at least one of the two
coalition distributions to move as prescribed. Even when the tested
coalitions must be disjoint, this requirement is incompatible with quota and
ex-ante proportionality from seven parties onward. The impossibility applies
to every such apportionment method, not only Sampford, and follows from support
constraints in a two-seat problem.

\begin{samepage}
\begin{definition}[Disjoint pairwise threshold monotonicity]
\label{def:disjoint-pairwise-threshold}
Let $v,v'\in\R_{\geq 0}^n$ be two vote vectors, each with positive total, let
$h\in\N$, and let $q,q'$ be their standard-quota vectors. An apportionment
method satisfies
\emph{pairwise threshold monotonicity for disjoint coalitions} if, for every
pair of disjoint coalitions $A,B\subseteq N$ satisfying
\begin{equation}
 q'_i\geq q_i\quad(i\in A),
 \qquad
 q'_j\leq q_j\quad(j\in B),
 \label{eq:disjoint-pairwise-changes}
\end{equation}
at least one of
\begin{equation}
 X_A(v',h)\succeq_{\SD}X_A(v,h)
 \qquad\text{or}\qquad
 X_B(v,h)\succeq_{\SD}X_B(v',h)
 \label{eq:disjoint-pairwise-branches}
\end{equation}
holds. The same comparison in \eqref{eq:disjoint-pairwise-branches} must hold
at every threshold; the selected branch cannot depend on the threshold.
\end{definition}
\end{samepage}

The word ``pairwise'' refers to a pair of coalitions, not to pairwise
correlation between parties. Definition~\ref{def:disjoint-pairwise-threshold}
is weaker than the unrestricted pairwise axiom of \citet{CGST+26a}, because it
tests only disjoint coalitions. It nevertheless contains ordinary threshold
monotonicity: take $B=N\setminus A$. Since $X_B=h-X_A$, the two alternatives
in \eqref{eq:disjoint-pairwise-branches} are then equivalent.

An apportionment method $\mathcal A$ has \emph{full support} if, for every
instance $(v,h)$ with $\sum_i v_i>0$ and every apportionment $x$ that gives
each party either its lower or upper quota, $\Pr[\mathcal A(v,h)=x]>0$.
Theorem~\ref{thm:pairwise-impossibility} does not require full support.

\begin{theorem}[Pairwise impossibility with or without full support]
\label{thm:pairwise-impossibility}
For every $n\geq7$, no randomized apportionment method on all strictly
positive vote profiles simultaneously satisfies quota, ex-ante
proportionality, and pairwise threshold monotonicity for disjoint coalitions.
The contradiction already uses house size two, quota vectors strictly between
zero and one, and a comparison in which every member of the first coalition
strictly gains while every member of the second strictly loses.
\end{theorem}

Set $h=2$ and let $q\in(0,1)^n$ satisfy $\sum_iq_i=2$. Quota then makes every
outcome a two-party set $S_q\in\binom{N}{2}$, and ex-ante proportionality
becomes
\begin{equation}
 \Pr_q[i\in S_q]=q_i\qquad(i\in N).
 \label{eq:pairwise-prescribed-marginals}
\end{equation}
Here and below, the subscript $q$ means that the method is evaluated at the
specific instance $(v,h)=(q,2)$; no homogeneity or quota-only assumption is
being made. For a pair $C=\{i,j\}$, since
$X_C=|S_q\cap C|\in\{0,1,2\}$,
\begin{equation}
 \Pr_q[X_C\geq2]=\Pr_q[S_q=C],
 \qquad
 \Pr_q[X_C\geq1]=q(C)-\Pr_q[S_q=C],
 \label{eq:pairwise-two-thresholds}
\end{equation}
where $q(C)=\sum_{i\in C}q_i$. Thus, when a pair's mean is fixed, an increase
in its joint-selection probability improves its threshold-two probability but
worsens its threshold-one probability.

This two-seat problem has a natural weighted-graph representation. The
parties are vertices, and each possible allocation $C\in\binom{N}{2}$ is an
edge with weight $\Pr_q[S_q=C]$. Equation
\eqref{eq:pairwise-prescribed-marginals} says that the total weight of the
edges incident to vertex $i$ is exactly $q_i$. After zero-weight edges are
removed, the remaining edges form the \emph{support graph}. We call an edge
family \emph{intersecting} when every two of its edges share a vertex.

The contradiction is encoded by the support graph on the small parties.
Pairwise threshold monotonicity forces this graph to be intersecting, but every
intersecting edge family lies in a star or a triangle, and neither form can
carry the prescribed weighted degrees. To force intersection, make party 1
almost certain. Jointly selecting either tested pair then becomes almost
impossible: the growing pair fails at threshold two, while the shrinking pair
fails in the opposite direction at threshold one.

\begin{samepage}
\begin{lemma}[Two disjoint small-party pairs cannot both occur]
\label{lem:two-disjoint-positive-pairs}
Suppose a quota and ex-ante proportional method satisfies
Definition~\ref{def:disjoint-pairwise-threshold}. Fix $n\geq7$ and consider
the two-seat quota profile
\begin{equation}
 p=\left(\frac35,
          \frac{7}{5(n-1)},\ldots,\frac{7}{5(n-1)}\right).
 \label{eq:pairwise-core-profile}
\end{equation}
Then the support graph induced by $\{2,\ldots,n\}$ is intersecting: there
cannot be two disjoint pairs $A,B\subseteq\{2,\ldots,n\}$ with
$\Pr_p[S_p=A]>0$ and $\Pr_p[S_p=B]>0$.
\end{lemma}
\end{samepage}

\begin{proof}
Put $m=n-1$ and $b=7/(5m)$, so that the profile in
\eqref{eq:pairwise-core-profile} is $(3/5,b,\ldots,b)$.
\begin{samepage}
Suppose otherwise. Choose
\[
 0<\varepsilon<\min\!\left\{
   \Pr_p[S_p=A],\frac{\Pr_p[S_p=B]}2,\frac12
 \right\},
\]
and then choose
\[
 0<\delta<\min\!\left\{
   b,1-b,\frac{\Pr_p[S_p=B]-\varepsilon}{2}
 \right\}.
\]
\end{samepage}
Let $R=\{2,\ldots,n\}\setminus(A\cup B)$ and put
\[
 \rho=\frac{1+\varepsilon-4b}{m-4}.
\]
Because $m\geq6$, we have $1-4b>0$ and $m-4\geq2$; the choice
$\varepsilon<1/2$ therefore gives $0<\rho<1$. The following table displays
the comparison between $p$ and the perturbed quota vector $p'$:
\begin{center}
\begin{tabular}{c@{\qquad}c@{\qquad}c}
\hline
party or block & quota under $p$ & quota under $p'$ \\
\hline
$1$ & $3/5$ & $1-\varepsilon$ \\
each $i\in A$ & $b$ & $b+\delta$ \\
each $i\in B$ & $b$ & $b-\delta$ \\
each $i\in R$ & $b$ & $\rho$ \\
\hline
\end{tabular}
\end{center}
All entries of $p'$ are strictly between zero and one and sum to two. Every
member of $A$ strictly gains and every member of $B$ strictly loses.

At the new profile, selecting either pair $A$ or pair $B$ requires omitting
party 1. Exact proportionality therefore gives
\begin{equation}
 \Pr_{p'}[S_{p'}=A],\ \Pr_{p'}[S_{p'}=B]
 \leq \Pr_{p'}[1\notin S_{p'}]=\varepsilon.
 \label{eq:pairwise-pair-upper-bound}
\end{equation}
The first branch of \eqref{eq:disjoint-pairwise-branches} fails at threshold
two because
\[
 \Pr_{p'}[X_A\geq2]=\Pr_{p'}[S_{p'}=A]
 \leq\varepsilon<\Pr_p[S_p=A]=\Pr_p[X_A\geq2].
\]
The second branch would require the old seat total of $B$ to dominate the new
one. But \eqref{eq:pairwise-two-thresholds} and
\eqref{eq:pairwise-pair-upper-bound} give
\begin{align*}
 \Pr_{p'}[X_B\geq1]
 &=2b-2\delta-\Pr_{p'}[S_{p'}=B]\\
 &\geq2b-2\delta-\varepsilon\\
 &>2b-\Pr_p[S_p=B]\\
 &=\Pr_p[X_B\geq1].
\end{align*}
Thus the second branch fails at threshold one, a contradiction.
\end{proof}

\begin{figure}[!t]
\centering
\begin{tikzpicture}[
  every node/.style={font=\footnotesize},
  v/.style={circle, draw=black!70, fill=white, inner sep=0pt, minimum size=5.5pt},
  big/.style={circle, draw=ravblue!50!black, fill=ravblue, inner sep=0pt, minimum size=8.5pt},
  hi/.style={circle, draw=black!70, fill=ravblue!18, inner sep=0pt, minimum size=5.5pt},
  e/.style={ravorange, line width=1.1pt},
  se/.style={ravorange, line width=0.9pt},
  possible/.style={ravgray!60, densely dashed, line width=0.5pt},
  lbl/.style={font=\scriptsize, text=ravgray},
  card/.style={fill=ravlight, draw=ravblue!15, line width=0.4pt, rounded corners=5pt},
  sub/.style={font=\scriptsize, text=black!80},
]
\begin{scope}
  \node[big, label={[text=ravblue!50!black]above:$1$}] (p1) at (1.25,1.8) {};
  \foreach \i/\x in {2/0,3/0.5,4/1,5/1.5,6/2,7/2.5}
     \node[v] (v\i) at (\x,0) {};
  \foreach \i in {2,...,7} \draw[possible] (p1) -- (v\i);
  \draw[e] (v2) -- (v3) node[midway, below=1pt, text=ravorange!85!black] {$A$};
  \draw[e] (v5) -- (v6) node[midway, below=1pt, text=ravorange!85!black] {$B$};
  \draw[decorate, decoration={brace, mirror, amplitude=3pt}, ravgray]
     (-0.15,-0.65) -- (2.65,-0.65) node[midway, below=4pt, text=black!80] {$L=\{2,\dots,n\}$};
  \node[lbl, anchor=west] at (2.9,1.8) {$q_1=\tfrac35$};
  \node[lbl, anchor=west] at (2.9,1.35) {possible edges at $1$: total $\tfrac35$};
  \node[lbl, anchor=west] at (2.9,0.0) {$q_i=\tfrac{7}{5(n-1)}$, $i\in L$};
  \node[lbl, anchor=west] at (2.9,-0.45) {edges inside $L$: total $\tfrac25$};
  \node[sub] at (3.525,-1.72) {(a) forbidden: disjoint positive edges $A,B$ inside $L$};
\end{scope}
\begin{scope}[yshift=-4.4cm]
  \node[hi, label={below:$c$}] (c) at (1.25,0) {};
  \foreach \i/\a in {2/150,3/110,4/70,5/30,6/-10}
     \node[v] (s\i) at ($(c)+(\a:1.15)$) {};
  \foreach \i in {2,...,6} \draw[se] (c) -- (s\i);
  \node[sub] at (1.275,-1.22) {(b) star:};
  \node[sub] at (1.275,-1.62) {$q_c\ge\tfrac25>\tfrac{7}{5(n-1)}$};
\end{scope}
\begin{scope}[xshift=4.5cm, yshift=-5.2cm]
  \node[hi, label={left:$a$}]  (a) at (0.6,0.2) {};
  \node[hi, label={right:$b$}] (b) at (1.9,0.2) {};
  \node[hi, label={above:$c$}] (c) at (1.25,1.25) {};
  \draw[se] (a) -- (b) -- (c) -- (a);
  \foreach \x/\y in {0/1.25,2.5/1.25,1.25/1.9}
     \node[v] at (\x,\y) {};
  \node[sub] at (1.275,-0.42) {(c) triangle:};
  \node[sub] at (1.275,-0.82) {$q_a+q_b+q_c\ge\tfrac45>\tfrac{21}{5(n-1)}$};
\end{scope}
\begin{scope}[on background layer]
  \draw[card] (-0.45,-1.30) rectangle (7.50,2.45);
  \draw[card] (-0.45,-5.35) rectangle (3.00,-2.95);
  \draw[card] (4.05,-5.35) rectangle (7.50,-2.95);
\end{scope}
\end{tikzpicture}
\caption{The two-seat support graph for $n=7$, where edge weights are
allocation probabilities and weighted degrees are quotas. By
Lemma~\ref{lem:two-disjoint-positive-pairs}, the positive edges inside
$L$ are intersecting and hence, by
Lemma~\ref{lem:intersecting-edge-families}, form a star or lie in a
triangle. Their total weight is $2/5$, so a star forces one marginal to
be at least $2/5$, while a triangle forces three marginals to sum to at
least $4/5$. Both contradict the prescribed small-party quotas. Dashed
lines denote possible edges incident to party~$1$.}
\label{fig:support}
\end{figure}

Lemma~\ref{lem:two-disjoint-positive-pairs} makes the small-party support graph
intersecting. Figure~\ref{fig:support} summarizes the forbidden configuration
and the two possible shapes of an intersecting edge family. The next elementary
lemma gives the classification.

\begin{lemma}[Intersecting edge families]
\label{lem:intersecting-edge-families}
If every two edges in a simple graph intersect, then all of its edges are
incident to one common vertex or all of its edges lie in one triangle.
\end{lemma}

\begin{proof}
The claim is immediate if the graph has at most one edge. Otherwise, take two
distinct intersecting edges and label them $ab$ and $ac$. Every edge not
incident to $a$ must intersect both of them and hence must be $bc$. If $bc$
is absent, every edge uses $a$. If $bc$ is present, any further edge must meet
all three of $ab$, $ac$, and $bc$, so it lies in the triangle on
$\{a,b,c\}$.
\end{proof}

\FloatBarrier
\begin{proof}[Proof of Theorem~\ref{thm:pairwise-impossibility}]
Fix $n\geq7$ and let $p$ be the profile in
\eqref{eq:pairwise-core-profile}. Since $\sum_ip_i=2$, the instance
$(v,h)=(p,2)$ has standard-quota vector $p$. Consider the support graph
introduced above, restricted to the small-party block
$L=\{2,\ldots,n\}$. Lemma~\ref{lem:two-disjoint-positive-pairs} says that
this graph is intersecting.

The total weight of its edges is
\[
 \Pr_p[S_p\subseteq L]
 =\Pr_p[1\notin S_p]
 =1-p_1=\frac25.
\]
By Lemma~\ref{lem:intersecting-edge-families}, the positive edges lie in a
star or a triangle. In the star case, the center is selected in every event
whose edge lies entirely in $L$, so its marginal is at least $2/5$. This
contradicts its prescribed marginal $7/[5(n-1)]<2/5$. In the triangle case,
every event counted among the small-party edges selects two of the triangle's
three vertices. Their marginal probabilities therefore sum to at least
$4/5$, whereas exact proportionality makes their sum
$21/[5(n-1)]<4/5$. This is again a contradiction.
\end{proof}

Consequently, the impossibility assumes neither continuity, neutrality, nor
full support. It applies to whatever positive support the method chooses: no
such support can realize the prescribed inclusion probabilities.

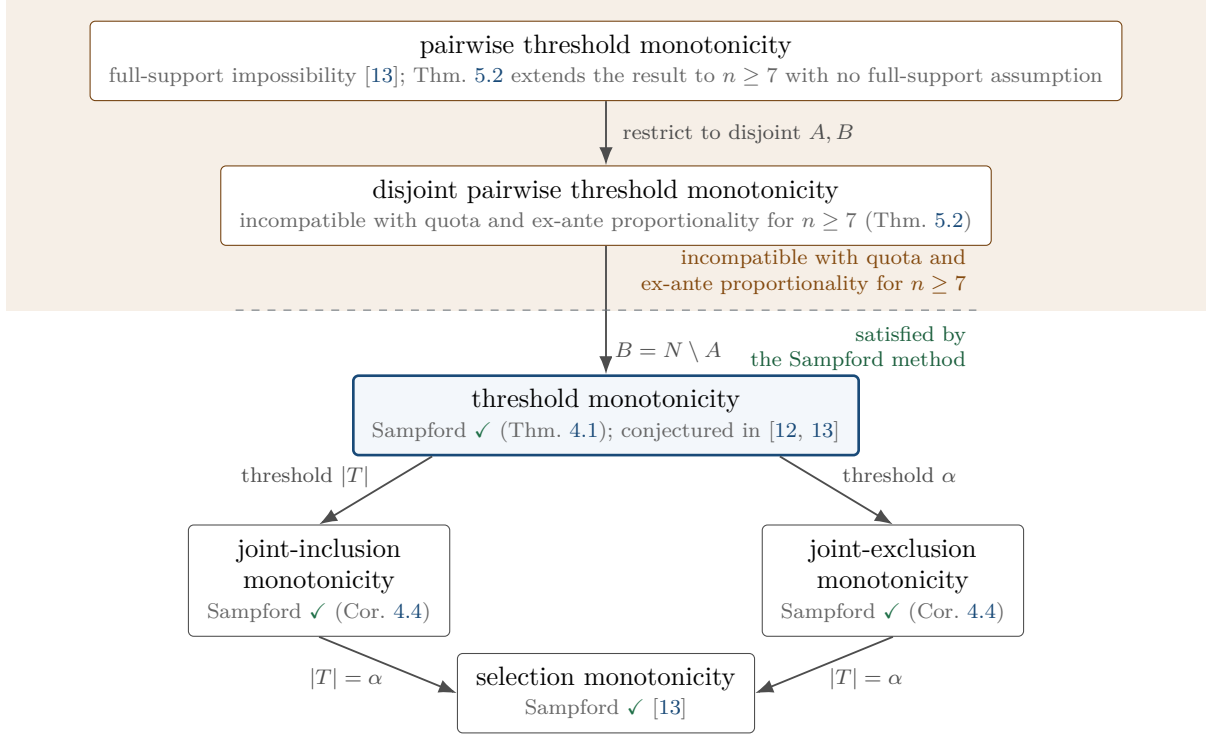
\begin{figure}[H]
\centering
\begin{tikzpicture}[
  ax/.style={draw=black!70, rounded corners=2pt, align=center,
             inner xsep=7pt, inner ysep=5pt, font=\small, fill=white},
  bad/.style={ax, draw=ravorange!70!black},
  main/.style={ax, draw=ravblue, line width=1pt, fill=ravlight},
  res/.style={font=\scriptsize, text=black!60, align=center},
  imp/.style={-{Latex[length=2.4mm]}, draw=black!70, line width=0.7pt},
  lab/.style={font=\scriptsize, text=black!70, fill=white, inner sep=1pt},
  node distance=8.5mm and 5mm,
]
\node[bad] (pw) {pairwise threshold monotonicity\\[-1pt]
  \textcolor{black!60}{\scriptsize
  full-support impossibility \citep{CGST+26a}; Thm.~\ref{thm:pairwise-impossibility}
  extends the result to $n\ge 7$ with no full-support assumption}};

\node[bad, below=of pw] (dpw) {disjoint pairwise threshold monotonicity\\[-1pt]
  \textcolor{black!60}{\scriptsize
  incompatible with quota and ex-ante proportionality for $n\ge 7$
  (Thm.~\ref{thm:pairwise-impossibility})}};

\node[main, below=17mm of dpw] (tm) {threshold monotonicity\\[-1pt]
  \textcolor{black!60}{\scriptsize
  Sampford \textcolor{ravgreen}{$\checkmark$} (Thm.~\ref{thm:sampford-threshold});
  conjectured in \citep{CGST+24a,CGST+26a}}};

\node[ax, below left=9mm and -13mm of tm] (ji) {joint-inclusion\\monotonicity\\[-1pt]
  \textcolor{black!60}{\scriptsize
  Sampford \textcolor{ravgreen}{$\checkmark$} (Cor.~\ref{cor:endpoint-events})}};

\node[ax, below right=9mm and -13mm of tm] (je) {joint-exclusion\\monotonicity\\[-1pt]
  \textcolor{black!60}{\scriptsize
  Sampford \textcolor{ravgreen}{$\checkmark$} (Cor.~\ref{cor:endpoint-events})}};

\node[ax, below=26mm of tm] (sel) {selection monotonicity\\[-1pt]
  \textcolor{black!60}{\scriptsize
  Sampford \textcolor{ravgreen}{$\checkmark$} \citep{CGST+26a}}};

\draw[imp] (pw) -- node[lab, fill=none, right=5pt]
  {restrict to disjoint $A,B$} (dpw);

\draw[imp] (dpw) -- node[lab, right=2pt, pos=0.82]
  {$B=N\setminus A$} (tm);

\draw[imp] ([xshift=-23mm]tm.south) -- coordinate[pos=0.48] (pji) (ji.north);

\draw[imp] ([xshift=23mm]tm.south) -- coordinate[pos=0.48] (pje) (je.north);

\node[lab, anchor=base east] at ([xshift=-2pt,yshift=2.5pt]pji |- pje)
  {threshold $|T|$};

\node[lab, anchor=base west] at ([xshift=2pt,yshift=2.5pt]pje)
  {threshold $\alpha$};

\draw[imp] (ji.south) -- node[lab, below left=1pt and -2pt, pos=0.45]
  {$|T|=\alpha$} (sel.west);

\draw[imp] (je.south) -- node[lab, below right=1pt and -2pt, pos=0.45]
  {$|T|=\alpha$} (sel.east);

\path (dpw.south) -- (tm.north) coordinate[midway] (mid);

\begin{scope}[on background layer]
  \fill[ravorange!10]
    ($(pw.north west)+(-11mm,3mm)$)
    rectangle
    ($(mid)+(49mm,0)$);
  \fill[ravorange!10]
    ($(mid)+(-49mm,0)$)
    rectangle
    ($(pw.north east)+(11mm,3mm)$);
\end{scope}

\draw[dashed, ravgray!90]
  ($(mid)+(-49mm,0)$) -- ($(mid)+(49mm,0)$);

\node[
  font=\scriptsize,
  text=ravorange!80!black,
  anchor=south east,
  align=right
] at ($(mid)+(49mm,0.8mm)$)
  {incompatible with quota and\\
   ex-ante proportionality for $n\ge 7$};

\node[
  font=\scriptsize,
  text=ravgreen!85!black,
  anchor=north east,
  align=right
] at ($(mid)+(49mm,-0.8mm)$)
  {satisfied by\\
   the Sampford method};

\end{tikzpicture}

\caption{Implications among the monotonicity axioms; arrows point from
stronger to weaker requirements and are labelled by the relevant
specialization. The shaded axioms are incompatible with quota and ex-ante
proportionality (Theorem~\ref{thm:pairwise-impossibility}),
whereas the Sampford method satisfies the unshaded requirements
(Theorem~\ref{thm:sampford-threshold},
Corollary~\ref{cor:endpoint-events}, and \citealp{CGST+26a}).
Joint-inclusion, joint-exclusion, and selection monotonicity concern the
residual selection lottery. The implications from threshold monotonicity
to the residual axioms use the standard embedding $q=p$ and
$h=\alpha$.}
\label{fig:axioms}
\end{figure}

\section{Conclusions}
\label{sec:discussion}

Our two results locate threshold monotonicity within randomized apportionment (Figure~\ref{fig:axioms}). The first is positive. Quota, ex-ante proportionality, and threshold monotonicity for arbitrary coalitions are compatible. Sampford sampling, a classical method, satisfies all three. This settles the conjecture of \citet{CGST+24a,CGST+26a}.  On the negative side, the pairwise strengthening fails from seven parties onward, even when restricted to disjoint coalitions. This impossibility applies whether or not the method has full support.

Our study on randomized apportionment sits at an exciting intersection of statistics and axiomatic social choice. 
We have barely scratched the surface of the axiomatic landscape. 
Pertinent to our results, several questions remain. For example, are there other natural rounding rules that satisfy threshold monotonicity? Finally, can threshold monotonicity be simultaneously achieved alongside house monotonicity?

\clearpage
\section*{Declaration on the use of generative AI}

Both results in the paper were derived via extensive and interactive use
of ChatGPT (Sol 5.6, OpenAI). The authors  verified all arguments; simplified, framed and
wrote the results; and take full responsibility for the manuscript.

\section*{Acknowledgments}

Haris Aziz and Simon Mackenzie are supported by the NSF-CSIRO
grant on ``Fair Sequential Collective Decision-Making'' (RG230833). Mashbat Suzuki is supported by the ARC Laureate Project FL200100204 on ``Trustworthy AI''.

\bibliographystyle{plainnat}
{\scriptsize
\raggedright
\setlength{\bibsep}{0pt}
\bibliography{randomized_apportionment_with_axioms}
}

\end{document}